\documentclass[12pt]{article}

\usepackage{arxiv}
\usepackage[utf8]{inputenc} 
\usepackage[T1]{fontenc}    
\usepackage{hyperref}    
\usepackage{url}            
\usepackage{booktabs}       
\usepackage{amsthm, amsmath, amsfonts, amssymb} 
\usepackage{color}

\newtheorem{theorem}{Theorem}

\usepackage{algpseudocode,algorithm2e,algorithmicx}

\usepackage{cleveref}  
\usepackage{nicefrac}       
\usepackage{microtype}      
\usepackage{graphicx}
\usepackage{doi}
\usepackage{color}
\usepackage{xcolor}
\usepackage[hyperref=true,doi=false,isbn=false,url=false,backend=biber,sorting=nyt,maxnames=99]{biblatex}
\theoremstyle{definition}
\newtheorem{example}{Example}

\title{More Efficient Identifiability Verification in\\ ODE Models by Reducing Non-Identifiability}

\author{
    Ilia Ilmer\\
    Ph.D. Program in Computer Science,\\
Graduate Center CUNY,\\
New York, NY, USA
    \And
Alexey Ovchinnikov\\
Ph.D. Program in Mathematics,\\
    CUNY Queens College,\\
    New York, NY, USA
\And
Gleb Pogudin,\\
LIX, CNRS, École Polytechnique,\\
Institute Polytechnique de Paris,\\
Paris, France
\And 
Pedro Soto\\
Ph.D. Program in Computer Science,\\
Graduate Center CUNY,\\
New York, NY, USA}

\date{}

\SetKwComment{Comment}{/* }{ */}
\RestyleAlgo{ruled}

\begin{document}
\maketitle
\begin{abstract}

    Structural global parameter identifiability indicates whether one can determine a parameter's value from given inputs and outputs in the absence of noise. If a given model has parameters for which there may be infinitely many values, such parameters are called non-identifiable. We present a procedure for accelerating a global identifiability query by eliminating algebraically independent non-identifiable parameters. Our proposed approach significantly improves performance across different computer algebra frameworks.

\end{abstract}

\section{Introduction}

Structural parameter identifiability is an crucial for design of mathematical models with ordinary differential equations (ODEs). For a given model, we may ask whether a parameter (or multiple parameters) can be discovered given sufficiently strong inputs and noiseless outputs. If the answer is positive, we say that a parameter (or multiple ones) is \emph{structurally} identifiable. We can further categorize structural identifiability into \emph{local} and \emph{global} identifiability. The former corresponds to multiple possible parameter values that can be recovered for a given model. The latter means that a parameter value can be recovered uniquely. Otherwise, we say that a parameter is non-identifiable.

Existing programs for parameter identifiability analysis rely on differential algebra and algebraic geometry. One example is SIAN \cite{hong_sian_2019,hong_global_2020}, implemented in {\sc Maple} \cite{hong_global_2020,ilmer_web-based_2021}, and Julia \cite{sian_julia_github}. It uses F4 algorithm \cite{faugere_new_1999} to compute Gr\"obner basis to determine global identifiability properties. This step can be computationally costly and has double exponential theoretical complexity in the worst case \cite{mayr_complexity_1982}. There have been further developments in Gr\"obner basis computation, see for instance \cite{bardet_complexity_2015,eder_survey_2017,faugere_new_2002}.

SIAN produces a polynomial system from the input ODE model. At that time in the program, the locally identifiable and non-identifiable parameters are known. Both of these classes of parameters are present in the polynomial system, however, if we know that some are non-identifiable, we can substitute several of them with numerical values thus reducing the workload for F4 algorithm.

In this work, we present a method for finding the combination of non-identifiable parameters that can be substituted in SIAN and significantly reduce the runtime of Gr\"obner basis algorithm. We will demonstrate the result of our algorithm on a collection of ODE models that may present a challenge for the F4 algorithm. The computation is performed in {\sc Maple} 2021.2 and Magma V2.26-8 on a computer cluster with 64 Intel Xeon CPUs with 2.30GHz clock frequency and 755 GB RAM.

\section{Related work}
\label{sec:relwork}

\subsection{Gr\"obner basis computation}

The original algorithm for finding a Gr\"obner basis of a polynomial ideal was presented by Buchberger in~\cite{buchberger_theoretical_1976}. However, the solution depends on multiple decisions such as selection strategy of polynomials and can be time-consuming~\cite{mayr_complexity_1982,giovini_one_1991}.
Faug\`ere presented F4~\cite{faugere_new_1999} and later F5~\cite{faugere_new_2002} algorithms that leverage better selection strategies for polynomials and linear algebra during computation. Recently,~\cite{bardet_complexity_2015} addressed the termination and complexity properties of the F5 algorithm. For an overview of F5-based solutions, see~\cite{eder_survey_2017}.

\subsection{Parameter identifiability}

Solutions for the identifiability problem have implementations in various programming languages. Structural Identifiability Analyzer (SIAN)  \cite{hong_sian_2019,hong_global_2020} was implemented in {\sc Maple} and Julia \cite{sian_julia_github} and is capable of global and local identifiability analysis. Algorithms for finding multi-experiment identifiable combinations \cite{ovchinnikov2020computing,ovchinnikov2020multi} extended SIAN and are available on the web \cite{ilmer_web-based_2021}. Fast local identifiability check based on power series was presented in \cite{sedoglavic2002probabilistic}. A new global identifiability algorithm of \cite{dong2022differential} is implemented in Julia and has been included into the Julia's Scientific Machine Learning (SciML) infrastructure. Among other widely used
packages, we highlight such solutions as web-based COMBOS~\cite{meshkat2014finding} and COMBOS 2 \cite{kalami2020combos2},
DAISY \cite{saccomani2008daisy} and DAISY 2 \cite{saccomani2019new} and GenSSI 2.0 \cite{ligon2018genssi}. For a deeper overview of existing identifiability methods, algorithms, software, and benchmarks we refer to \cite{miao2011identifiability,chis2011structural,villaverde2019benchmarking,villaverde2019observability,raue2014comparison}.

\section{Main result}
\label{sec:main_result}

\subsection{Preliminary information}

We aim to accelerate global identifiability assessment of SIAN~\cite{hong_sian_2019}. Let us present the typical input format accepted by the program
\begin{equation}~\mathbf{\Sigma} :=~\begin{cases}
        \mathbf{x}' & =\mathbf{f}(\mathbf{x},~\boldsymbol\mu,~\mathbf{u}),~ \\
        \mathbf{y}  & =\mathbf{g}(\mathbf{x},~\boldsymbol\mu,~\mathbf{u}).
    \end{cases}
    \label{eq:ode}
\end{equation}
where~\(\mathbf{f}=(f_1,\dots, f_n)\) and~\(\mathbf{g}=(g_1,\dots,g_n)\) with~\(f_i=f_i(\mathbf{x},~\boldsymbol\mu,~\mathbf{u})\),~\(g_i=g_i(\mathbf{x},~\boldsymbol\mu,~\mathbf{u})\) are rational functions over the field of complex numbers~\(\mathbb{C}\).

The vector~\(\mathbf{x}=(x_1,\dots,x_n)\) represents the time-dependent state variables and~\(\mathbf{x}'\) represents their derivatives. The time-dependent vector-function~\(\mathbf{u}=(u_1,\dots,u_s)\) represents the input variables. The~\(m\)-vector~\(\mathbf{y}=(y_1,\dots,y_n)\) represents the output variables. The vector~\(\boldsymbol\mu=(\mu_1,\dots,\mu_{\lambda})\) represents constant parameters and~\(\mathbf{x}^\ast=(x_1^\ast,\dots,x_n^\ast)\) defines initial conditions of the model.

The output functions \(\mathbf{y}\) are differentiated to compute truncated Taylor polynomials at time \(t=0\), see \cite[Theorems 3.16, 4.12]{hong_global_2020} for details on the truncation bound. Further, SIAN samples \(\mathbf{x}^*,\boldsymbol{\mu}\) to evaluate each component \(y_i\) and its derivatives.
Gr\"obner basis lets one check if the sampled quantities \(\mathbf{x}^*,\boldsymbol{\mu}\) are unique that result in the computed values \(y_i\). 
This makes SIAN a randomized Monte-Carlo algorithm, and the
user can specify the probability of correctness.

\subsection{Finding transcendence basis}

At the time of computing Gr\"obner basis in SIAN, we know the local identifiability and non-identifiability of all parameters. It is tempting to exclude the latter
from further consideration by numerical substitution. However, there is a subtlety, as the choice among non-identifiable parameters may affect the identifiability properties of others.
\begin{example}
    \label{example1}
    Consider the following ODE system
    \begin{equation}
        \begin{cases}
            \dot{x}_1 = -p_1x_1 + p_2x_2 + u, & \dot{x}_2 = p_3x_1 - p_4x_2 + p_5x_3, \\
            \dot{x}_3 = p_6x_1 - p_7x_3,      & y = x_1.
        \end{cases}
    \end{equation}

    The Structural Identifiability Toolbox \cite{ilmer_web-based_2021} reports that parameter \(p_6\) as non-identifiable. Assume we would like to substitute \(p_6\Rightarrow 0\) everywhere in the ODE.
    As a result,
    only-locally-identifiable parameters \(p_4, p_7\) become globally identifiable. If we substitute non-zero numbers into all non-identifiable parameters, for instance, \(p_2\Rightarrow 131,~p_3\Rightarrow 93,~p_5\Rightarrow17,~p_6\Rightarrow 41\), the outcome is the same.
\end{example}

To increase the efficiency by maximizing our substitution choices, we perform substitutions into a maximal set of algebraically independent variables  \(B\) (a transcendence basis) in the polynomial system $E^t$ produced in  \cite[Algorithm~1, Step~2]{hong_global_2020}. Such a set is found in \Cref{algo:alg_indep}.
In the implementation,  the Jacobian matrix  entries are sampled randomly according to   \cite[Algorithm~1, Step~3]{hong_global_2020}
to find pivot columns. We can use the same sampling bound as in \cite{hong_global_2020} to get the pivot columns with the same probability of correctness.

\LinesNumbered{}
\begin{algorithm}
    \caption{Finding algebraically independent parameters}\label{algo:alg_indep}
    \KwIn{Polynomial system \(E^t\) in variables $X$
    }
    \KwOut{
        Set of algebraically independent variables
    }
    \(JacobianMatrix \gets \frac{\partial E^t}{\partial X};\)\\
    \(JacobianMatrix.sample();\)\\
    \(Pivots \gets JacobianMatrix.pivotColumns();\)\\
    \KwRet{\(\{x \mid
        x\not\in Pivots\}\)}
\end{algorithm}

\begin{theorem}\label{thm:main}
    Under the substitution of random uniformly distributed integers from $[1,\frac{4}{3}D_2]$ (where $D_2$ is defined in \cite[Algorithm~1]{hong_global_2020}) into the variables in $B$, the  probability $p$ of correctness of \cite[Algorithm~1]{hong_global_2020} is still guaranteed.
\end{theorem}

\begin{proof} See Section~\ref{sec:proof}
\end{proof}

\subsection{Heuristics for best choice}

The result of \Cref{algo:alg_indep} is not unique and can change based on the column arrangement. For a basis of size \(k\), we considered all \(N=\binom{n}{k}\) combinations of columns if feasible. For large values \(N\), we create a large sample of \(K<N\) combinations. In the final code, the user can specify sample size \(K\).
The heuristic for picking the best possible transcendence basis is as follows:
\begin{enumerate}
    \item For a given transcendence basis \(T\) of size \(k\), before performing the substitution, collect degrees of monomials (i.e. sum of degrees of each variable in monomial) that contain elements of \(T\) avoiding double-counting of degrees (i.e., if a monomial contains more than one transcendental element, we collect its degree once);
    \item From the previous step, we obtained an array \(A_{\mu}\) of integer degree values \(d_{\mu}[m]\) for each member \(\mu \in T\) and monomial \(m\);
    \item Normalize the array by computing \(A_{\mu}[m] :=\frac{d_{\mu}[m]}{\sum\limits_{d_{\mu}[m'] \in A_{\mu}}d_{\mu}[m']}\)
    \item Compute the entropy of \(A_{\mu}\) as
          \(\mathcal{H}_{\mu} := -\sum_{m}A_{\mu}[m]\log\left(A_{\mu}[m]\right).\)
    \item At this step, we have a collection of entropies \(\{\mathcal{H}_{\mu} | \mu \in T\}\) for each possible transcendence basis. Sort this collection, so that larger entries end up last (this can be done by a lexicographic approach). Pick the last transcendence basis.
\end{enumerate}
\Cref{heuristic_scatter} shows how the CPU time changes depending on the value of the resulting entropies. We pick the right-most basis with "highest" entropy values. Upper bound trend line (in red) shows that our choice yields better-than-median improvement.
\begin{figure}
    \centering
    \includegraphics[width=0.5\textwidth]{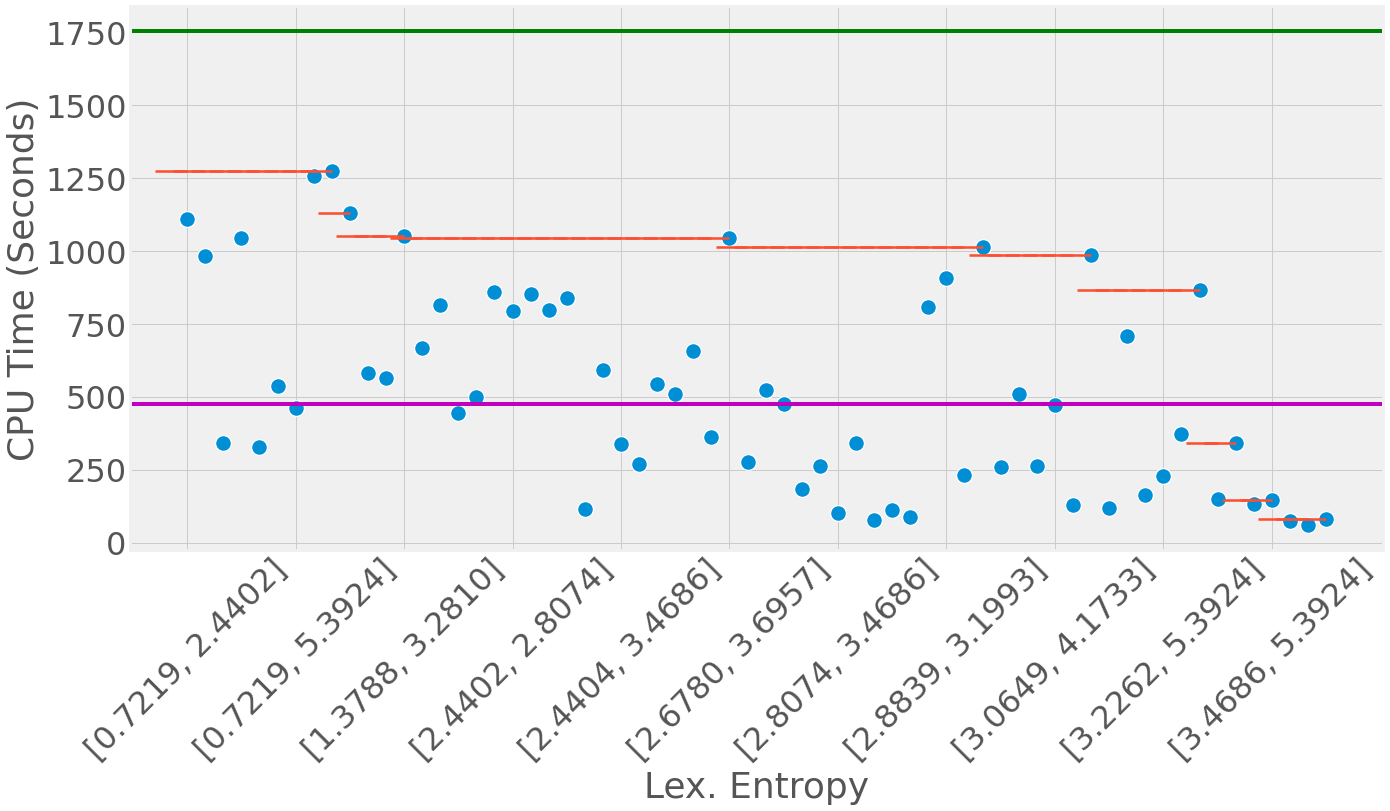}~~\vspace{0.5em}
    \includegraphics[width=0.5\textwidth]{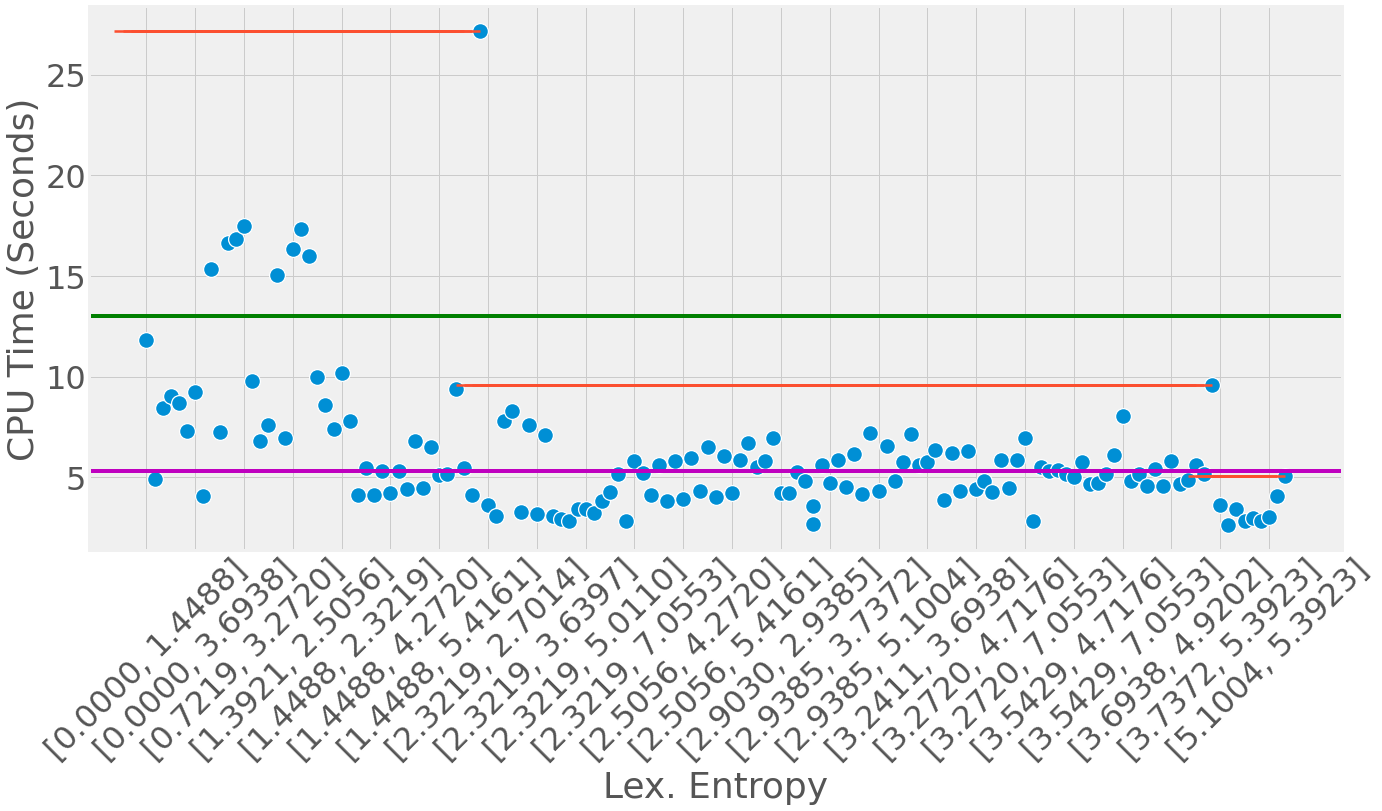}
    \caption{Impact of entropy-based transcendence basis choice on the CPU time of F4 algorithm for two ODE models. The {\color{green} green} line denotes the default time (original system), the {\color{magenta} magenta} line shows the median CPU time of each basis, the {\color{red} red} lines show the upper bound trend. Top image: \Cref{goodwin}, bottom image: \Cref{HIV}.}
    \label{heuristic_scatter}
\end{figure}
\subsection{Why does maximum degree weighted count entropy work?}

Let us explain the reasoning behind the approach above.
The \emph{degree weighted count entropy} of a transcendental element $\mu$ is capturing three heuristics at once:
\begin{enumerate}
    \item Does $\mu$ appear in a lot of monomials?
    \item Does $\mu$ have a large total degree?
    \item Does $\mu$ \textit{typically} contribute to monomials of a large degree?
\end{enumerate}
All considerations above arise from possible causes of computational hardness for an F4 algorithm to compute Gr\"obner basis.
Simply put, we are counting the instances of $\mu$ with a weight that captures its \textit{typical} (see \cite{cover1999elements})
degree weight.
Let us illustrate this with a simple example: suppose that \(p_1\) and \(p_2\) appear in the following collections \(M_{p_1},~M_{p_2}\)
\begin{align*}
    M_{p_1} & = \{ p_1^{10} x_1^5x_2^5 , p_1x_1 , p_1 x_2, p_1 x_3,  p_1 x_4, p_1 x_5 \},                  \\
    M_{p_2} & = \{ p_2^{2} x_1^2x_2 , p_1^3x_1^2 , p_1 ^2x_2^3, p_1 ^3x_3^2,  p_1^2 x_4^3, p_1^3 x_5^2 \}.
\end{align*}
Both sets
have the same total degree and number of monomials, 30 and 6, respectively. However, $p_1$ is in one complicated monomial, but all monomials of $p_2$ are complicated. Taking the degree weighted count entropy, we get
\begin{gather*}
    \mathcal{H}_{p_1} = -\frac{10}{30} \log_2\left(\frac{10}{30}\right) - 5 \frac{2}{30} \log_2\left(\frac{2}{30}\right) \approx 1.831,\\
    \mathcal{H}_{p_2} = -6 \frac{5}{30} \log_2\left(\frac{5}{30}\right) \approx 2.584,
\end{gather*}
and the maximum entropy method selects the more complicated transcendental element in cases where simply counting degrees or occurrences would fall apart.

\section{Benchmarks}
In this section, we present CPU time and memory usage comparison across three different setups for {\sc Maple} and Magma computer algebra systems. We use the variable order as in \cite[eq. 8]{bessonov2022obtaining}. We present CPU time and memory usage of F4 algorithm on SIAN-produced ideals:
\begin{itemize}
    \item with no changes (positive dimension)
    \item without transcendence basis (zero-dimensional ideals), denoted as ``0-dim''
    \item without transcendence basis and with weighted ordering as defined by the main result of \cite{bessonov2022obtaining}
\end{itemize}
These results are shown in \cref{benchmarks_0dim_maple,benchmarks_0dim_weights_maple,benchmarks_0dim_magma,benchmarks_0dim_weights_magma}.
We observe significant improvements, especially in cases where {\sc Maple} would otherwise be unable to complete the Gr\"obner basis computation. Combined with the weighted ordering \cite{bessonov2022obtaining}, we observe an significant combined speedup. In \Cref{benchmarks_0dim_maple} and \Cref{benchmarks_0dim_weights_maple}, ``N/A'' stands for the following error message returned by {\sc Maple}:
``\texttt{Error, (in Groebner:-F4:-GroebnerBasis) numeric exception:~division by zero}''.
\vspace{0.2em}
\begin{table*}[!h]
    \centering
    \resizebox{\columnwidth}{!}{
        \begin{tabular}{||c|c|c|c||c|c|c||c|c|c||}
            \hline
            \multicolumn{4}{||c||}{Model information} & \multicolumn{3}{c||}{Time (min)} & \multicolumn{3}{c||}{Memory (GB)}                                                                        \\
            \hline
            Model                                     & num.                             & num.                              & tr. & default & 0-dim  & speedup    & default & 0-dim  & improvement \\
            name                                      & polys.                           & vars.                             & deg & system  & system &            & system  & system &             \\\hline\hline
            COVID Model,                              &                                  &                                   &     &         &        &            &         &        &             \\
            \cref{ssaair}                             & 49                               & 48                                & 2   & N/A     & N/A    & N/A        & N/A     & N/A    & N/A         \\\hline
            QWWC                                      &                                  &                                   &     &         &        &            &         &        &             \\
            \cref{qwwc}                               & 58                               & 50                                & 1   & N/A     & 111.8  & \(\infty\) & N/A     & 1.21   & \(\infty\)  \\\hline
            SIR COVID Model                           &                                  &                                   &     &         &        &            &         &        &             \\
            \cref{siraqj}                             & 79                               & 81                                & 7   & 12      & 0.1    & 13.4       & 11.5    & 0.8    & 14.1        \\\hline
            Goodwin Oscillator                        &                                  &                                   &     &         &        &            &         &        &             \\
            \cref{goodwin}                            & 42                               & 43                                & 2   & 29.8    & 1.3    & 22.4       & 10.6    & 0.8    & 12.5        \\\hline
            SEIR,                                     &                                  &                                   &     &         &        &            &         &        &             \\
            \cref{seir}                               & 44                               & 45                                & 2   & 2.2     & 0.4    & 5.9        & 3.3     & 0.5    & 6.2         \\\hline
            HIV,                                      &                                  &                                   &     &         &        &            &         &        &             \\
            \cref{HIV}                                & 59                               & 55                                & 2   & 0.2     & <0.1   & 4.3        & 0.2     & 0.1    & 2.4         \\\hline
        \end{tabular}
    }
    \caption{Results of Gr\"obner basis computation step of SIAN with positive characteristic \(p=11863279\) in {\sc Maple} 2021.2. We show comparison of default computation and zero-dimensional system (without transcendence basis).}
    \label{benchmarks_0dim_maple}
\end{table*}
\begin{table*}[!ht]
    \centering
    \resizebox{\columnwidth}{!}{
        \begin{tabular}{||c|c|c|c||c|c|c||c|c|c||}
            \hline
            \multicolumn{4}{||c||}{Model information} & \multicolumn{3}{c||}{Time (min)} & \multicolumn{3}{c||}{Memory (GB)}                                                                                                                                                                           \\
            \hline
            Model                                     & num.                             & num.                              & tr. & weights                      & 0-dim with                           & speedup & weights                      & 0-dim with                           & improvement \\
            name                                      & polys.                           & vars.                             & deg & \cite{bessonov2022obtaining} & weights \cite{bessonov2022obtaining} &         & \cite{bessonov2022obtaining} & weights \cite{bessonov2022obtaining} &             \\\hline\hline
            COVID Model,                              &                                  &                                   &     &                              &                                      &         &                              &                                      &             \\
            \cref{ssaair}                             & 49                               & 48                                & 2   & 602.3                        & 340.6                                & 1.7     & 23.2                         & 20.6                                 & 1.1         \\\hline
            QWWC                                      &                                  &                                   &     &                              &                                      &         &                              &                                      &             \\
            \cref{qwwc}                               & 58                               & 50                                & 1   & 2.5                          & 1.8                                  & 1.1     & 1.4                          & 0.6                                  & 1.7         \\\hline
            SIR COVID Model                           &                                  &                                   &     &                              &                                      &         &                              &                                      &             \\
            \cref{siraqj}                             & 79                               & 81                                & 7   & 51.1                         & 5                                    & 10.3    & 10.7                         & 1.7                                  & 6.2         \\\hline
            Goodwin Oscillator                        &                                  &                                   &     &                              &                                      &         &                              &                                      &             \\
            \cref{goodwin}                            & 42                               & 43                                & 2   & 1.6                          & 0.8                                  & 1.9     & 0.7                          & 0.5                                  & 1.4         \\\hline
            SEIR,                                     &                                  &                                   &     &                              &                                      &         &                              &                                      &             \\
            \cref{seir}                               & 44                               & 45                                & 2   & 0.1                          & 0.1                                  & 1.0     & 0.1                          & 0.1                                  & 1.0         \\\hline
            HIV,                                      &                                  &                                   &     &                              &                                      &         &                              &                                      &             \\
            \cref{HIV}                                & 59                               & 55                                & 2   & 0.1                          & < 0.1                                & 3.3     & 0.1                          & <0.1                                 & 2.9         \\\hline
        \end{tabular}
    }
    \caption{Results of Gr\"obner basis computation step of SIAN with positive characteristic \(p=11863279\) in {\sc Maple} 2021.2. We show comparison of weighted ordering in positive- and zero-dimensional systems.}
    \label{benchmarks_0dim_weights_maple}
\end{table*}
\begin{table*}[!ht]
    \centering
    \resizebox{\columnwidth}{!}{
        \begin{tabular}{||c|c|c|c||c|c|c||c|c|c||}
            \hline
            \multicolumn{4}{||c||}{Model information} & \multicolumn{3}{c||}{Time (min)} & \multicolumn{3}{c||}{Memory (GB)}                                                                     \\
            \hline
            Model                                     & num.                             & num.                              & tr. & default & 0-dim  & speedup & default & 0-dim  & improvement \\
            name                                      & polys.                           & vars.                             & deg & system  & system &         & system  & system &             \\\hline\hline
            COVID Model,                              &                                  &                                   &     &         &        &         &         &        &             \\
            \cref{ssaair}                             & 49                               & 48                                & 2   & 3471.2  & 55.1   & 63.0    & 36.4    & 27.1   & 1.3         \\\hline
            QWWC                                      &                                  &                                   &     &         &        &         &         &        &             \\
            \cref{qwwc}                               & 58                               & 50                                & 1   & 429.0   & 296.1  & 1.4     & 11.0    & 6.5    & 1.7         \\\hline
            SIR COVID Model                           &                                  &                                   &     &         &        &         &         &        &             \\
            \cref{siraqj}                             & 79                               & 81                                & 7   & 6.6     & 0.6    & 10.7    & 5.6     & 0.7    & 8.2         \\\hline
            Goodwin Oscillator                        &                                  &                                   &     &         &        &         &         &        &             \\
            \cref{goodwin}                            & 42                               & 43                                & 2   & 22.4    & 1.5    & 14.7    & 3.1     & 0.5    & 5.7         \\\hline
            SEIR,                                     &                                  &                                   &     &         &        &         &         &        &             \\
            \cref{seir}                               & 44                               & 45                                & 2   & 3.8     & 0.5    & 7.8     & 2.0     & 0.3    & 6.0         \\\hline
            HIV,                                      &                                  &                                   &     &         &        &         &         &        &             \\
            \cref{HIV}                                & 59                               & 55                                & 2   & <0.1    & <0.1   & 3.2     & 0.3     & 0.2    & 1.5         \\\hline
        \end{tabular}
    }
    \caption{Comparison of Gr\"obner basis computation step of SIAN in Magma 2.26-8. We specify a positive characteristic \(p=11863279\). We show comparison of default computation and zero-dimensional system (without transcendence basis).}
    \label{benchmarks_0dim_magma}
\end{table*}
\begin{table*}[!ht]
    \centering
    \resizebox{\columnwidth}{!}{
        \begin{tabular}{||c|c|c|c||c|c|c||c|c|c||}
            \hline
            \multicolumn{4}{||c||}{Model information} & \multicolumn{3}{c||}{Time (min)} & \multicolumn{3}{c||}{Memory (GB)}                                                                                                                                                                           \\
            \hline
            Model                                     & num.                             & num.                              & tr. & weights                      & 0-dim with                           & speedup & weights                      & 0-dim                                & improvement \\
            name                                      & polys.                           & vars.                             & deg & \cite{bessonov2022obtaining} & weights \cite{bessonov2022obtaining} &         & \cite{bessonov2022obtaining} & weights \cite{bessonov2022obtaining} &             \\\hline\hline
            COVID Model,                              &                                  &                                   &     &                              &                                      &         &                              &                                      &             \\
            \cref{ssaair}                             & 49                               & 48                                & 2   & 517.4                        & 334.4                                & 1.5     & 21.6                         & 11.5                                 & 1.9         \\\hline
            QWWC                                      &                                  &                                   &     &                              &                                      &         &                              &                                      &             \\
            \cref{qwwc}                               & 58                               & 50                                & 1   & 1.2                          & 1.2                                  & 1.0     & 1.7                          & 0.8                                  & 2.2         \\\hline
            SIR COVID Model                           &                                  &                                   &     &                              &                                      &         &                              &                                      &             \\
            \cref{siraqj}                             & 79                               & 81                                & 7   & 31.4                         & 3.6                                  & 8.7     & 18.9                         & 3.8                                  & 5.0         \\\hline
            Goodwin Oscillator                        &                                  &                                   &     &                              &                                      &         &                              &                                      &             \\
            \cref{goodwin}                            & 42                               & 43                                & 2   & 0.8                          & 0.6                                  & 1.3     & 0.5                          & 1.0                                  & 0.5         \\\hline
            SEIR,                                     &                                  &                                   &     &                              &                                      &         &                              &                                      &             \\
            \cref{seir}                               & 44                               & 45                                & 2   & < 0.1                        & < 0.1                                & 0.8     & 0.2                          & 0.2                                  & 1.2         \\\hline
            HIV,                                      &                                  &                                   &     &                              &                                      &         &                              &                                      &             \\
            \cref{HIV}                                & 59                               & 55                                & 2   & < 0.1                        & < 0.1                                & 3.6     & 0.3                          & < 0.1                                & 4.1         \\\hline
        \end{tabular}
    }
    \caption{Comparison of Gr\"obner basis computation step of SIAN in Magma 2.26-8. We specify a positive characteristic \(p=11863279\). We show comparison of weighted ordering with and and without transcendence basis.}
    \label{benchmarks_0dim_weights_magma}
\end{table*}

\section{Conclusions}
We presented an enhancement that decreases runtime and memory consumption of identifiability testing algorithm SIAN \cite{hong_sian_2019}. Our solution
consists of finding and removing the transcendence basis from the algebraically independent non-identifiable model parameters. In addition, the method can be of increased
efficiency if combined with other Gr\"obner basis-specific improvements, such as weighted variable orderings. We showed
this by applying the weighted ordering method  \cite{bessonov2022obtaining} to a polynomial system with transcendence basis removed.

It is possible to optimize the substitution procedure by, for instance, sampling the random values for transcendental elements directly into the polynomials. However, we did not observe significant differences at the final Gr\"obner basis stage.

\printbibliography{}

\appendix{}

\section{ODE models}
In this section, we  prove our main result about the probability guarantee for the random substitutions into the transcendence basis and present details about ODE models we considered for transcendence basis substitution. We will present the ODEs, the output functions, and the discovered bases of the models used in the analysis of this paper.

\subsection{Proof of Theorem~\ref{thm:main}}\label{sec:proof}
By \cite[Lemma~4.4]{hong_global_2020}, applied to $X$ being an irreducible component $V_i$  of the variety $V$ defined by $\widehat{E^t}$ and $\pi$ the projection to the affine space $\mathbb{A}^{|B|}$ of  the $B$-variables, there exists a proper subvariety $Y_i \subset\mathbb{A}^{|B|}$
such that $\deg Y_i \leq \deg V_i$ and, for every
$a \in  \mathbb{A}^{|B|}\setminus Y_i$,
we have: $\pi^{-1}(a)\cap V_i \ne \varnothing$.
Suppose $\theta$ is a locally but not globally identifiable parameter.
By \cite[Theorem~5.5]{hong_global_2020}, the projection of $V$ onto $\theta$-axis is finite and has the cardinality $> 1$ if the integer substitutions that are used to convert $E^t$ into $\widehat{E^t}$ are outside the zero set of the polynomials $Q$ and $\widetilde{P_2}$, which are defined in the proof of \cite[Theorem~5.5]{hong_global_2020}.
Thus there exist components, say, $V_1$ and $V_2$ of $V$ with different projections onto the $\theta$-axis.
We choose $b$ such that $b \not\in Y_i$ for every $i$.
Then there exist $a_1 \in V_1 \cap \pi^{-1}(b)$ and $a_2 \in V_2 \cap \pi^{-1}(b)$ having different $\theta$-coordinates.
By \cite[Proposition~3]{Heintz}, there exists a nonzero polynomial $P_s$ of degree at most
$\sum_i\deg Y_i\leq\deg V$
such that, for each $i$, $P_s$ is zero on $Y_i$.
So, outside of the zero set of $P_sQ\widetilde{P_2}$, substitutions into the $B$-variables cannot convert non-global identifiability into global identifiability.
The proof of \cite[Theorem~5.5]{hong_global_2020} introduces a number $D_2\geq 6\deg \widehat{ E^t}/(1-p)\geq 6\deg P_s/(1-p)$ and such that $\deg Q\widetilde P_2 \leq D_2(1-p)/2$.
We finally have $\deg P_sQ\widetilde P_2\leq D_2(1-p)(\frac{1}{2}+\frac{1}{6})=D_2(1-p)\frac{4}{6}  =\left(\frac{4}{3}D_2\right)(1-p)/2$ to obtain the desired probability guarantee.

\subsection{Goodwin oscillator}

The system \cref{goodwin} originates in~\cite{goodwin_oscillatory_1965} describing time periodicity in cell behavior. This example has 4 state variables \(x_{1,2,3,4}\) and 6 parameters.
\begin{equation}
    \begin{cases}
        \dot{x}_1 = -b   x_1 + \frac{1}{(c + x_4)}, & \dot{x}_2 = \alpha   x_1 - \beta   x_2,                               \\
        \dot{x}_3 = \gamma   x_2 - \delta   x_3,    & \dot{x}_4 = \frac{\sigma   x_4   (\gamma   x_2 - \delta   x_3)}{x_3}, \\
        y = x_1
    \end{cases}
    \label{goodwin}
\end{equation}
The transcendence basis found by our approach here has size 2 and the optimal combination consists of the initial condition \(x_3(0)\) and parameter \(\gamma\).

\subsection{A different SEIR-like COVID-19 model}

The following model is a COVID-19 epidemiological model coming from ~\cite[example 37, table 1]{massonis_structural_2020}.

\begin{equation}
    \begin{cases}
        \dot{S}_d =  -\epsilon_s \beta_a (A_n + \epsilon_a A_d) S_d - h_1 S_d + h_2 S_n - \epsilon_s \beta_i S_d I_n, \\
        \dot{S}_n =  -\beta_i S_n I_n - \beta_a (A_n + \epsilon_a A_d) S_n + h_1 S_d - h_2 S_n,                       \\
        \dot{A}_d =  \epsilon_s \beta_i S_d I_n + \epsilon_s \beta_a (A_n + \epsilon_a A_d) S_n + h_2 A_n -           \\
        - \gamma_{ai} A_d - h_1 A_d,                                                                                  \\
        \dot{A}_n =  \beta_i S_n I_n + \beta_a (A_n + \epsilon_a A_d) S_n + h_1 A_d -                                 \\
        - \gamma_{ai} A_n - h_2 A_n,                                                                                  \\
        \dot{I}_n =  f \gamma_{ai} (A_d + A_n) - \delta I_n - \gamma_{ir} I_n,                                        \\
        \dot{R}   =  (1-f) \gamma_{ai} (A_d + A_n) + \gamma_{ir} I_n,                                                 \\
        y_1  = S_d,  y_2  = I_n
    \end{cases}
    \label{ssaair}
\end{equation}
In this model, even with transcendental elements removed, the computation does not finish without using weighted ordering via algorithm of \cite{bessonov2022obtaining}. The transcendence basis is \(\delta,~R(0)\).

\subsection{HIV epidemic model}
This is a biomedical model applied to HIV infection in~\cite{hiv2_paper}.
The outputs were changed to make the system more of a computational challenge to SIAN.
\begin{equation}
    \begin{cases}
        \dot{x} = \lambda - d x - \beta x v,\dot{y} = \beta x v - a y, \\
        \dot{v} = k y - u v, \dot{w} = c z y w - c q y w - b w,        \\
        \dot{z} = c q y w - h z,                                       \\
        y_1 = w,  y_2 = z
    \end{cases}
    \label{HIV}
\end{equation}
The transcendence basis of this model is \(\beta, c\).

\subsection{SEIR epidemiological model of COVID-19}

The next SEIR model for COVID-19 was presented in~\cite[Example 34]{massonis_structural_2020}.
\begin{equation}
    \begin{cases}
        \dot{S} = \Lambda - r  \beta S  I / N - \mu  S, & \dot{E} = \beta  S  I / N - (\epsilon  + \mu)  E, \\
        \dot{I} = \epsilon  E -(\gamma + \mu)  I,       & \dot{R} = \gamma  I - \mu  R,                     \\
        y = I + R.
    \end{cases}
    \label{seir}
\end{equation}

Transcendence basis for \Cref{seir} is \(\beta, N\).

\subsection{QWWC model}

This model comes from~\cite[Equation 67]{harrington2017reduction}
\begin{equation}
    \label{qwwc}
    \begin{cases}
        \dot{x} = -xa + zy + ay, \dot{w} = ez - wf + xy, \\
        \dot{z} = -cz - wd + xy, \dot{y} = bx + by - xz, \\
        g = x
    \end{cases}
\end{equation}

In this model, there is only one transcendental parameter \(d\). Substituting this parameter, we observe a tremendous speedup: the system finishes computation without error in {\sc Maple}.

\subsection{SIR-like COVID model}
This COVID-19 model has transcendence degree 7 with basis consisting of \(A(0), I(0), N(0), R(0), d_2,d_3,d_6\). The model comes from \cite[26]{massonis_structural_2020} where we added the equation \(\dot{N}=0\).
\begin{equation}
    \label{siraqj}
    \begin{cases}
        \dot{S} = b\,N - S\,(I\,\lambda + \lambda\,Q\,\epsilon_a\,\epsilon_q + \lambda\,\epsilon_a\,A + \lambda\,\epsilon_j\,J + d_1), \\
        \dot{I} = k_1\,A - (g_1 + \mu_2 + d_2)\,I,                                                                                     \\
        \dot{R} = g_1\,I + g_2\,J - d_3\,R,                                                                                            \\
        \dot{A} = S\,\lambda(I + \epsilon_a\,\epsilon_q Q+ \epsilon_a A + \epsilon_j J) - (k_1 + \mu_1 + d_4)\,A,                      \\
        \dot{Q} = \mu_1\,A - (k_2 + d_5)\,Q,                                                                                           \\
        \dot{J} = k_2\,Q + \mu_2\,I - (g_2 + d_6)\,J,                                                                                  \\
        \dot{N} = 0,
        y_1 = Q,
        y_2 =J
    \end{cases}
\end{equation}
\thanks{The authors are grateful to CCiS at CUNY Queens College. This work was partially supported by the NSF under grants CCF-1563942, CCF-1564132, DMS-1760448, DMS-1853650, and DMS-1853482}
\end{document}